\documentclass[11pt]{article}

\usepackage[utf8]{inputenc}
\usepackage[T1]{fontenc}
\usepackage[english]{babel}
\usepackage{lmodern}
\usepackage{xcolor}

\usepackage[colorlinks=False]{hyperref}  % https://en.wikibooks.org/wiki/LaTeX/Hyperlinks

\usepackage{comment}

\usepackage{amsmath}  % use [intlimits]?
\usepackage{amsfonts,dsfont}
\usepackage{amssymb}
\usepackage{amsthm}
\usepackage{mathtools}
\usepackage{thmtools}
\usepackage{bm} % bold math symbols, especially greek

\numberwithin{equation}{section} % equation number: (s.n) where s is section number, n is counter

\newlength{\bredde}
\def\slash#1{\settowidth{\bredde}{$#1$}\ifmmode\,\raisebox{.15ex}{/}
\hspace*{-\bredde} #1\else$\,\raisebox{.15ex}{/}\hspace*{-\bredde} #1$\fi}
\DeclarePairedDelimiter{\absolute}{\lvert}{\rvert} % absolute value |x|
\renewcommand{\epsilon}{\varepsilon} % replace the ugly epsilon with the nice epsilon
\renewcommand{\subset}{\subseteq} % more explicit subset symbol
\let\Re\relax\DeclareMathOperator{\Re}{Re} % redeclare Re (real part of a complex number) as math operator
\newcommand*{\iunit}{\mathrm{i}} % imaginary unit (i^2 = -1)
\newcommand*{\Nset}{\mathds{N}} % natural numbers
\newcommand*{\Zset}{\mathds{Z}} % whole numbers
\newcommand*{\Rset}{\mathds{R}} % real numbers
\newcommand*{\Cset}{\mathds{C}} % complex numbers

\newcommand*{\GL}{\mathrm{GL}} % general linear group
\newcommand*{\detij}{\underset{1\leq i,j\leq N}{\det}}
\newcommand*{\detjk}{\underset{1\leq j,k\leq N}{\det}}

\usepackage{extpfeil}
\newextarrow{\xrightrightarrows}{{5}{8}{0}{0}}
{\bigRelbar\bigRelbar{\bigtwoarrowsleft\rightarrow\rightarrow}}

\DeclareMathOperator{\Trace}{Tr} % trace of a matrix
\DeclareMathOperator{\Pf}{Pf} % pfaffian of a skew-symmetric matrix
\newcommand*{\abs}[1]{\absolute*{#1}} % the shortest way to make |x|
\declaretheorem[numberwithin=section]{proposition}
\declaretheorem[numberlike=proposition]{theorem}
\declaretheorem[numberlike=proposition]{lemma}
\declaretheorem[numberlike=proposition]{corollary}

\makeatletter
\newcommand*{\doublerightarrow}[2]{\mathrel{
  \settowidth{\@tempdima}{$\scriptstyle#1$}
  \settowidth{\@tempdimb}{$\scriptstyle#2$}
  \ifdim\@tempdimb>\@tempdima \@tempdima=\@tempdimb\fi
  \mathop{\vcenter{
    \offinterlineskip\ialign{\hbox to\dimexpr\@tempdima+1em{##}\cr
    \rightarrowfill\cr\noalign{\kern.5ex}
    \rightarrowfill\cr}}}\limits^{\!#1}_{\!#2}}}
\makeatother

\begin{document}
\title{Generalised unitary group integrals of Ingham-Siegel and Fisher-Hartwig type
}
\author{{\sc Gernot Akemann$^1$, Noah Ayg\"un$^1$, and Tim R. W\"urfel$^2$}\\~\\
$^1$Faculty of Physics, Bielefeld University, PO-Box 100131, D-33501 Bielefeld, Germany\\
$^2$Department of Mathematics, King's College London, London WC2R 2LS, UK
}
%\keywords{}

\date{}

\maketitle

\begin{abstract}

We generalise well-known integrals of Ingham-Siegel and Fisher-Hartwig type over the unitary group $U(N)$ with respect to Haar measure, for finite $N$ and including fixed external matrices. When depending only on the eigenvalues of the unitary matrix, they can be related to a Toeplitz determinant with jump singularities. After introducing fixed deterministic matrices as external sources, the integrals can no longer be solved using Andr\'ei\'ef's integration formula. Resorting to the character expansion as put forward by Balantekin, we derive explicit determinantal formulae containing Kummer's confluent and Gau{\ss}' hypergeometric function. They depend only on the eigenvalues of the deterministic matrices and are analytic in the parameters of the jump singularities.
Furthermore, unitary two-matrix integrals of the same type are proposed and solved in the same manner. When making part of the deterministic matrices random and integrating over them, we obtain similar formulae  in terms of Pfaffian determinants. This is  reminiscent to a unitary  group integral found recently by Kanazawa and Kieburg.

\end{abstract}

%%%%%%%%%%%%%%%%%%%%%%%%%%%%%%%%%%%%%%%%%%%%
\section{Introduction and Main Results}\label{intro}

Integrals over the unitary group $U(N)$ with respect to Haar measure appear in many different areas in physics and mathematics. Examples include two-dimensional lattice gauge theory, the Gross--Witten--Wadia model \cite{GW,Wadia}, the epsilon-regime of chiral perturbation theory \cite{GL}, or  the Ising model via Toeplitz determinants, see \cite{DIK} for a recent mathematical review. In mathematics they appear in combinatorics in the longest monotone subsequence of random permutations \cite{TW} and in number theory in the statistics of the non-trivial zeros of the Riemann zeta-function \cite{KS}. 
In physics terminology, such group integrals correspond to the partition function.
Therefore, a detailed knowledge of such integrals is very useful at finite-$N$, in order to study the scaling behaviour as $N\to\infty$, the analyticity at fixed $N$ and to study phase transitions, see \cite{GW,Wadia,Miguel}.
Because the matrix elements of $U\in U(N)$ are randomly distributed, such group integrals can be also viewed as random matrix ensembles, as introduced by Dyson under the name of circular ensembles \cite{DysonI}. 

In the case when the integrand of unitary group integrals only depends on the eigenvalues $t_j$, $j=1,\ldots N$, of $U\in U(N)$ and factorises, the Andr\'ei\'ef integration formula \cite{An} can be applied to map the group integral to a determinant of a Toeplitz matrix. This follows from the Jacobian of the diagonalisation in terms of the modulus square of the Vandermonde determinant of the eigenvalues \cite{DysonI}. This Toeplitz matrix is given by a single integral, with the symbol containing the initial integrand. 

In many applications, however, the integrand may also depend on a fixed deterministic matrix, and thus on the eigenvectors of $U\in U(N)$. Such fixed matrices can parametrise  relevant physical information, e.g. different quark masses in the chiral Lagrangian in the Leutwyler-Smilga integral \cite{LS} or the effect of temperature in terms of Matsubara frequencies in the corresponding random matrix model \cite{Andy,Tilo}. 
In this case different techniques have to be applied. 
A powerful method has been proposed by Balantekin \cite{Baha}, to use an expansion in terms of characters of irreducible representations of the unitary group (or the general linear group $\GL(N,\mathbb{C})$), and to exploit their orthogonality, together with the Cauchy--Binet identity. Examples for such group integrals are the celebrated Harish-Chandra--Itzykson-Zuber integral 
\cite{HC,IZ}
and the Berezin-Karpelevich integral \cite{BK}, see \cite{SW03} for their most general form using the character expansion, which has inspired us.
The character expansion has been generalised to the orthogonal and symplectic group \cite{Baha2} and to unitary supergroups \cite{LWGW}.

Although depending on fixed external matrices, taken from $\GL(N,\mathbb{C})$ in general, also these integrals have integrands that are given in terms of invariants, that is in terms of determinants and (typically exponentials of) traces, or both. In this work we will present several  such unitary group integrals where only partial results were known in the literature, \cite{Yan,FK1,FK2}, see also \cite{KKS15} for a similar integral with constrained matrices.
We will generalise known integrals of Fisher-Hartwig type as they appear in the study of the Ising model and we propose a unitary version of the Ingham-Siegel integral, to contain external matrices in both cases. Without external matrices the latter integral has previously appeared in the study of the longest monotone subsequence of random partitions \cite{TW}. 

Most unitary group integrals mentioned so far take a determinantal form. 
Difficulties arise, however, when higher powers of matrices or traces arise, e.g. $(\Trace[U])^2$ or $\Trace[U^2]$,  as they appear when adding Wilson terms, describing the effect of finite-lattice spacing in the chiral Lagrangian \cite{DSV}. They can be dealt with using a scalar \cite{ADSV}, respectively matrix-valued Hubbard-Stratonovic transformation \cite{KK18}. In the latter case this leads to a closed form with a Pfaffian determinant. Based on the generalised two-matrix Ingham-Siegel and Fisher-Hartwig integrals that we find below, we can construct further examples for unitary group integrals with a Pfaffian structure. They are obtained when part of the fixed external matrices become random, too, and are integrated over.

In all the following we denote by $dU_N$ the Haar measure of the unitary group $U(N)$ and use the normalisation $\int dU_N=1$. The Vandermonde determinant that appears frequently is defined as 
\begin{equation}
\Delta_N( t ) := \Delta_N(t_1, \cdots t_N) = \det\left[ t_j^{N-k}\right]_{j,k=1}^N = 
(-1)^{\frac{N(N-1)}{2}}\det\left[ t_j^{k-1}\right]_{j,k=1}^N =
\prod_{1 \leq j < k \leq N} \left( t_j -t_k \right) . 
\label{Vdet}
\end{equation}
Notice that different conventions exist, without the additional sign in the middle formula, when swapping the order in each factor under the product. For completeness, we recall the definition of two standard special functions, the  confluent Kummer hypergeometric function
\begin{equation}
{}_1F_1 \left(a ; b ; z \right) := 
 \sum_{m=0}^{\infty}\frac{\Gamma \left( a +m \right) \Gamma \left( b \right)}{\Gamma \left( a \right) \Gamma \left( b+m \right)} \frac{z^{m}}{m!} ,
  \label{1F1DEF}
\end{equation}
and Gau{\ss} hypergeometric function 
\begin{equation}
 _2F_1(a,b;c;z):= 
 \sum_{m=0}^\infty \frac{\Gamma \left( c \right) \Gamma \left(a +m \right)\Gamma \left( b+m \right)}{\Gamma \left( a \right)\Gamma \left( b \right)\Gamma \left(c+m \right) } \frac{z^m}{m!}. 
\label{2F1Def}
\end{equation}
Both are defined when $b$, respectively $c$ are not non-positive integers. 
The series for ${}_1F_1(a ; b ; z )$ converges everywhere for $z\in\mathbb{C}$, is analytic in $a$ and $z$, and meromorphic in $b$. For $b>0$ fixed as we will encounter below, 
${}_1F_1(a;b;z)$ is thus analytic.
The series for ${}_2F_1(a,b;c;z)$ converges inside the unit disc, $|z|<1$, and converges absolutely on the unit circle $|z|=1$ for $\Re(c-a-b)>0$. The principal branch of 
${}_2F_1(a,b;c;z)/\Gamma(c)$ (with respect to $c$) is 
analytic in $a,b,c$ for $z\setminus\{0,1,\infty\}$. 
Because below we will have $c>0$ fixed,  ${}_2F_1(a,b;c;z)$ is analytic in the remaining variables. 
For more details we refer to \cite{NIST}. 

The matrix norm used below for matrices $A\in{\rm GL}(N,\mathbb{C})$ is defined as usual by the vector norm as $||A||:=\sup\{||Ax||: x\in\mathbb{C}^N \mbox{with}\ ||x||=1\}$.

%%%%%%%%%%%%%%%%%%%%%%%%%%%%%%%%%%%%%
\begin{theorem}[Unitary Ingham-Siegel Integrals]
\label{Thm-IS}
Let $A,B,C,D\in{\rm GL}(N,\mathbb{C})$ be fixed complex, invertible $N\times N$ matrices and denote by $\mu_j^2$ and $\nu_j^2$, $j=1,\ldots,N$, the eigenvalues of $AD$ and $BC$, respectively. 
Without loss of generality we can assume that $||C||,||D||<1$, and let $\alpha\in\mathbb{C}$ be fixed.
Then, the following two integrals are well defined and we have the following identities:
\begin{eqnarray}
Z_\alpha^{\rm IS}( A, D ) &:=& \int  dU_N \det \left[1+ U^\dagger D\right]^\alpha \exp \left[ \Trace (AU) \right]
\nonumber\\
&=& 
\frac{ \detjk\left[ _1F_1 (-\alpha ; N-j+1 ; - \mu^2_k)\, \mu_k^{2(N-j)} \right]
}{\Delta_N(\mu^2)},
\label{1IS}
\end{eqnarray}
and 
\begin{eqnarray}
Z_\alpha^{\rm IS}( A, B,C,D ) &:=& \int  dU_N\!\int dV_N \det \left[1+ V^\dag CU^\dagger D\right]^\alpha \exp \left[ \Trace (AUBV) \right]
\nonumber\\
&=& g_N(\alpha)
\frac{ \detjk\Big[ {}_1F_1 (-\alpha -N+1 ; 1; -\mu_j^2\nu_k^2)\Big]}{\Delta_N(\mu^2)\Delta_N(\nu^2)},
\label{2IS}
\end{eqnarray}
with $g_N(\alpha):=\prod_{j=1}^{N}\frac{(j-1)!^2\Gamma(\alpha+N-j+1)}{\Gamma(\alpha+N)}$.
\end{theorem}
In the case when the integration over $U(N)$ is replaced by positive Hermitian matrices, $A$ is negative and $\alpha=m$ is a positive integer, such integrals were introduced by Ingham \cite{Ingham} and Siegel \cite{Siegel}. For $m$ a negative integer such Ingham-Siegel integrals were used in \cite{Yan} as a tool to compute negative moments in the Fisher-Hartwig integral \eqref{sFHresult} for $\alpha=\beta=m$, cf. Theorem~\ref{Thm-GFH} below.

Both integrals \eqref{1IS} and \eqref{2IS} can be interpreted as hypergeometric functions of matrix argument(s), see \cite{Orlov}.
Note that there is an equivalent result for the first integral \eqref{1IS}, see \eqref{1ISalt} below, where the first index of ${}_1F_1$ is shifted to $-\alpha-j+1$. 

In the case when $AD$ has real eigenvalues and $A$ and $D$ are normal, \eqref{1IS} is real. 
This can be seen when using the invariance of the Haar measure $dU_N$ under the transformation $U\to U^\dag$, and interchanging $A\leftrightarrow D$, which leaves the first integral invariant. Because of the properties of $A$ and $D$ the eigenvalues of $AD$ agree with those of $DA$ (and are real). 
For the second integral, under the same conditions on $B$ and $C$, transforming $V\to V^\dag$ and swapping $B\leftrightarrow C$ in addition, the integral \eqref{2IS} is manifestly real. 

When in the first integral \eqref{1IS} both fixed matrices are proportional to the identity, $A=b\mathbf{1}_N$, $b\in\mathbb{C}$ and $D=\pm \mathbf{1}_N$, we obtain as a  corollary the integrals in \cite{TW} (after transforming $U\to U^\dag$). Because $||D||=1$ in this case, we have to restrict the possible range of $\alpha$, such that the integral on the left-hand side exists.
%%%%%%%%%%%%%%%%%%%%%%
\begin{corollary}\label{Cor-TW-IS}
Let $\alpha,b\in\mathbb{C}$, with $\Re(\alpha)>-1$. Then the following identity holds:
\begin{eqnarray}
Z_{\alpha,\pm}^{\rm IS}(b) &:=& \int  dU_N \det \left[1\pm U^\dagger \right]^\alpha \exp \left[ b\Trace (U) \right]
\nonumber\\
&=& \frac{1}{(2\pi)^N}\detjk\left[\binom{\alpha}{j-k}{}_1F_1(-\alpha+j-k;j-k+1;\mp b)\right].
\label{1IS'}
\end{eqnarray}
\end{corollary}
In \cite{TW}[Theorem 1] these integrals appear as generating functions for $b\in\mathbb{R}$, in order to determine the lengths of the longest increasing respectively 
decreasing subsequences in words, with signs $\pm$ and $\alpha=\pm k$, $k\in\mathbb{N}$ in the symbol of the Toeplitz-determinant. There, these integrals enjoy an alternative representation in terms of Painlev\'e V \cite{TW}. Notice that due to the definition \eqref{1F1DEF} the matrix inside the determinant in \eqref{1IS'} is non-vanishing for $j<k$, as the  diverging Gamma-functions cancel.

\begin{theorem}[Generalised Fisher-Hartwig Integrals]
\label{Thm-GFH}
Let $A,B,C,D\in{\rm GL}(N,\mathbb{C})$ be fixed complex, invertible $N\times N$ matrices and denote by $\mu_j^2$ and $\nu_j^2$, $j=1,\ldots,N$, the eigenvalues of $AD$ and of $BC$, respectively. In the case that all norms are restricted, $||A||,||B||,||C||,||D||<1$, there is no restriction on $\alpha,\beta\in\mathbb{C}$. Else we have to require $\Re(\alpha+\beta)>-1$. Then the following two integrals are well defined and obey the identities  
\begin{eqnarray}
Z_{\alpha,\beta}^{\rm FH}( A, D ) &:=& \int  dU_N \det \Big[1+ AU\Big]^\alpha\det \left[1+ U^\dagger D\right]^\beta
\nonumber\\
&=& \frac{ \detjk\left[ _2F_1 (-\alpha-j+1;-\beta -j+1; N-j+1 ;  \mu_k^2)\, \mu_k^{2(N-j)} \right] 
}{\Delta_N(\mu^2)},
\label{1FH}
\end{eqnarray}
and 
\begin{eqnarray}
Z_{\alpha,\beta}^{\rm FH}( A, B, C, D ) &:=& \int  dU_N\int dV_N \det \Big[1+ AUBV\Big]^\alpha\det \left[1+ V^\dag CU^\dagger D\right]^\beta
\nonumber\\
&=& \frac{g_N(\alpha) g_N(\beta)}{\prod_{j=1}^{N} (j-1)!^2}
\frac{ \detjk\left[ _2F_1 (-\alpha-N+1;-\beta-N+1 ; 1 ;  \mu_j^2\nu_k^2)\right] 
}{\Delta_N(\mu^2)\Delta_N(\nu^2)}.\quad\quad
\label{2FH}
\end{eqnarray}
\end{theorem}
Once again both integrals can be seen as hypergeometric functions of matrix argument(s) \cite{Orlov}. A similar partition function also appears in the context of the sum of two Wishart ensembles with unequal covariances \cite{Kumar}.
The first integral \eqref{1FH} has appeared previously in the literature \cite{FK1} for $\alpha=\beta=m$, where $m$ is a positive or negative integer. There, different representations were presented, in terms of determinants of size $m$. For  $\alpha=m$ and $\beta=n$ positive integers such integrals have been related to integrals over Jacobi ensembles of size $\min(m,n)$ in \cite{FK2}. Such relations are an instance of more general relations between matrix integrals and were called Colour-Flavour transformation by Zirnbauer \cite{Martin}. For general ratios of integer powers of such determinants averaged over the unitary group, with external matrices being proportional to the identity, we refer to \cite{CFZ,CFN}.
Notice that an equivalent representation to \eqref{1FH} exist, see \eqref{1FH'}, with $-\beta-j+1\to-\beta$ in the second argument of ${}_2F_1$.\\

If we choose $B=C=X$ to be a Hermitian random matrix and integrate over it with respect to the flat Lebesgue measure $dH_N$ for complex Hermitian matrices, we obtain two examples for group integrals with a Pfaffian structure, rather than a determinant. 
This structure follows closely from \cite{KKS15}, where such a matrix integral results form a matrix Hubbard-Stratonovic transformation. Applied to Theorem~\ref{Thm-IS},  \eqref{2IS} we obtain the following:
%%%%%%%%%%%%%%%%%%%%%%%%%%%%%%%%%%%%%%%%%%%%%%
\begin{corollary}[Integrated Unitary Ingham-Siegel Integral]\label{cor:JIS}
Let $ A,D \in {\rm GL}(N,\mathbb{C})$ be fixed, with eigenvalues $a_j^2$, $j=1,\ldots,N$ of $DA$,  and $X=X^\dag$ be a complex Hermitian random matrix distributed according to $\exp[-\Trace v(X)]$, where $v$ decays sufficiently fast at infinity. Furthermore, $\alpha\in \Cset$ is a fixed constant with $\Re (\alpha)>-1$. Then the following integral is obtained:
\begin{eqnarray}
J_{\alpha}^{\rm IS}(A,D)
&:=& \int dH_N e^{-\Trace v(X)}\int  dU_N\!\int dV_N \det \left[1+ V^\dag XU^\dagger D\right]^\alpha \exp \left[ \Trace (AUXV) \right]
\nonumber\\
&=&c_N\frac{N!  g_N(\alpha)}{\Delta_N(a^2)}  \left\{
\begin{array}{ll}
\Pf \left[E_\alpha^{\rm IS}(a_j,a_k) \right]_{j,k=1}^N &, N \text{ even}, \\[2mm]
\Pf\left[ \begin{array}{ c | c }
E_\alpha^{\rm IS}(a_j,a_k) & F_\alpha^{\rm IS}(a_j) \\\hline
-F_\alpha^{\rm IS}(a_k) & 0 \\
\end{array} \right]_{j,k=1}^N &, N \text{ odd}, \\
\end{array}
\right.
\label{JSPfaff}
\end{eqnarray}
where $c_N=\frac{\pi^{\frac{N(N-1)}{2}}}{\prod_{j=1}^N j!}$, and we denote
\begin{eqnarray}
 E_\alpha^{\rm IS}(a_j,a_k) &:=& \frac12\int_{\Rset^2} dxdy\frac{x-y}{x+y} e^{-v(x)-v(y)} 
 \nonumber  \\ 
 &&\times \,
 \bigg( {}_1F_1 (1-\alpha-N+1;1;-a^2_jx^2) {}_1F_1 (1-\alpha-N+1; 1 ;- a^2_k y^2) \nonumber  \\ 
&&\quad- {}_1F_1 (1-\alpha-N+1;1;-a^2_kx^2) {}_1F_1 (1-\alpha-N+1; 1 ; -a^2_j y^2)  \bigg),
 \label{EISajak}
\end{eqnarray}
and 
\begin{equation}
 F_\alpha^{\rm IS}(a_j) := \int_{\Rset} dx e^{-v(x)} {}_1F_1 (1-\alpha -N+1; 1 ; -a^2_j x^2)  . 
 \label{FISaj}
\end{equation}
\end{corollary}
A similar result holds when integrating the generalised Fisher-Hartwig integral \eqref{2FH} in Theorem~\ref{Thm-GFH} over a random matrix $X=B=C$, see 
Corollary~\ref{cor:JFH} in Section \ref{proofs}. It is also possible of course to take the one-matrix integrals \eqref{1IS} or \eqref{1FH}, make $A=D=X$ random and integrate over $X$ in the same way. This will then no longer depend on an external matrix, and such results can also be obtained from Corollary~\ref{cor:JIS} and Corollary~\ref{cor:JFH}, in taking the degenerate limit $a_k=1$ for $k=1,\ldots,N$. We will not give more details here.  

The remainder of this paper is organised as follows. In Section \ref{charev}
we will review elements of the character expansion of the unitary group as needed to prove our theorems. This is largely based on \cite{Baha}. In particular, we will prove a Lemma providing three equivalent formulae for the dimension of an irreducible representation of the unitary group needed later. Section \ref{proofs} is devoted to the proofs of our main Theorems~\ref{Thm-IS} and \ref{Thm-GFH} and Corollaries~\ref{Cor-TW-IS} and \ref{cor:JIS}, in that order. After concluding and discussing open problems in Section \ref{conclusio}, we provide some useful identities for binomial coefficients and Gau{\ss}' hypergeometric function in Appendix \ref{appA}. In Appendix \ref{appB} the standard Fisher-Hartwig integral without external matrices is presented to which we compare as a limiting case, together with Andr\'ei\'ef's integral identity.

%%%%%%%%%%%%%%%%%%%%%%%%%%%%%%%%%%%%%%%%%%%%
\section{The Character Expansion Revisited}\label{charev}

In this section, we will recall the character expansion of the unitary group $U(N)$ and more generally ${\rm GL}(N,\mathbb{C})$, following closely \cite{Baha}. As an important extension we will provide three different formulae for the dimension of a representation in Corollary~\ref{CorolaDimR} below.  
An extension to other compact groups can be found in \cite{Baha2}. The material presented here will be used in the proofs of the main theorems in Section \ref{proofs}.

An irreducible representation $r=(n_1,\ldots,n_N)$ of the unitary group $U(N)$ can be labeled by an ordered partition of non negative integers $n_1\geq \ldots \geq n_N\geq 0$. Then, as shown by Weyl in \cite{Weyl}, the character $\chi_r(U)$ of representation $r=(n_1,\ldots,n_N)$ of $U\in U(N)$ is given by
\begin{align}
  \chi_{r}(U)=\chi_{(n_1,\dots,n_N)}(U)=\frac{\detjk \left[ t_j^{n_k+N-k} \right] }{\Delta_N \left( t\right) }=\frac{\detjk \left[ t_j^{n_k+N-k} \right] }{\detjk \left[ t_j^{N-k} \right] },
\label{WCF}
\end{align}
called Weyl character formula (for $U(N)$).
If we have such a character expansion, we can use the Schur orthogonality
\begin{align}
  \int  dU_N \; U^{(r)}_{jk} U^{*(r')}_{lm} = \frac{1}{d_r} \delta^{rr'} \delta_{jl} \delta_{km},
\label{shurOg}
\end{align}
where  $U^{(r)}_{jk}$ are the components of $U$ in the representation $r$, and $d_r = \chi_{r}(\mathbf{1}_{N} )$ is the dimension of the representation, the character of the identity. In particular, since $\chi_r(U) = \sum_{j=1}^{d_r} U^{(r)}_{jj}$, we have the following orthogonality of characters, 
\begin{align}
  \int  dU_N \; \chi_r(U) \chi_{r'}(U^\dagger ) = \delta^{rr'},
\label{ChOg}
\end{align}
that will allow us to do the integration over the unitary group. 

The irreducible representation $r$ of the general linear group $\GL(N,\mathbb{C})$ is similar to $U(N)$ and  can also be labeled by the partition $r=(n_1,\cdots,n_N)$. In particular, Weyl's character formula \eqref{WCF} equally holds.
Therefore, $\chi_r(AU)=\sum_{j,k=1}^{d_r} A^{(r)}_{j,k }U^{(r)}_{k,j}$, where $A\in \GL(N,\mathbb{C})$ and $U\in U(N) \subset \GL(N,\mathbb{C})$, can be seen as a character of $\GL (N,\mathbb{C})$ and we can generalise \eqref{ChOg} through \eqref{shurOg}. We obtain
\begin{align}
\int dU_N \chi_r \left( B U \right) \chi_{r'} ( A U^\dagger ) = \delta^{rr'} \frac{1}{d_r}\chi_r \left(AB\right),
\label{MatOg}
\end{align}
with $A,B \in \GL (N,\mathbb{C})$. This step will allow us to do the group integrals in the presence of fixed matrices $A,B$, too. In the sequel, it will be important to be able to cancel the factors $d_r$ in the denominator. For that purpose, different explicit forms given in the following corollary will be most helpful.
%%%%%%%%%%%%%%%%%%%%%%%%%%%%%%%%%%
\begin{corollary}[Dimension Formulae]
\label{CorolaDimR}
The dimension $d_r$ of a given representation $r = \left( n_1, \cdots n_N \right)$ can be written in three equivalent forms as
\begin{align}
d_r &= \frac{\Delta_N\left(n_1-1, \ldots, n_N-N \right)}{\prod_{j=1}^{N} (N-j)!},
\label{dimensionR}\\
d_r &= \detjk \left[ \frac{1}{\Gamma \left( n_k+j-k+1 \right)} \right] \prod_{j=1}^N \frac{(n_j+N-j)!}{(N-j)!}, \label{DrSub1} \\ 
d_r &= \detjk \left[ \binom{\alpha}{n_k+j-k}\right] \prod_{j=1}^N  \frac{(n_j+N-j)! \Gamma ( \alpha - n_j+j )}{(N-j)! \Gamma (\alpha +N-j+1) } ,
 \label{DrSub2}
\end{align} 
with $\alpha \in \Cset$ being fixed and arbitrary. 
\end{corollary}
Notice that because the Gamma-function has  first order poles at non positive integers, for such $\alpha$ we have singularities in the determinant \eqref{DrSub2}. However, they will cancel with zeros in the denominator, making the relation well defined. For binomial coefficients with general $\alpha$ see Appendix \ref{appA}.

\begin{proof}
The first form \eqref{dimensionR} is standard and can be found in textbooks, see e.g. \cite{Weyl}. Alternatively, it is not difficult to derive it by applying the rule of l'H\^opital to the Weyl character \eqref{WCF}, in the limit when all eigenvalues degenerate to $t_j\to1$, for $j=1,\ldots,N$.

We will first prove the equivalence between \eqref{dimensionR} and \eqref{DrSub1}. The idea is to rewrite the determinant of the Gamma-function, such that we arrive at the Vandermonde determinant in \eqref{dimensionR}. 
Defining 
\begin{equation}
m_j=n_j+N-j, 
\label{mjdef}
\end{equation}
the Gamma-function can be rewritten as
\begin{align}
\frac{1}{\Gamma(m_k-N+j+1)} &= \frac{(m_k-N+j+1)\cdots (m_k-1)m_k}{\Gamma\left( m_k +1 \right)} \nonumber\\
&= \frac{m_k^{N-j}+\mathcal O (m_k^{N-j-1})}{m_k!}.
\end{align}
In the determinant the lower order terms in the numerator 
can be subtracted by using the invariance of the determinant. We obtain
\begin{equation}
  \detjk \left[ \frac{1}{\Gamma \left( m_k-N+j+1 \right)} \right] 
  =\detjk \left[ m_k^{N-j} \right] \prod_{k=1}^N \frac{1}{m_k!},
\label{drcorrola1}
\end{equation}
as the factorials can be taken out of the determinant.
The right-hand side is proportional to the Vandermonde determinant  $\Delta_N(m_1,\ldots,m_N)$. Using its product form in \eqref{Vdet}, the variables $m_k$ can be shifted by $N$ to $m_k-N=n_k-k$ inside the determinant. This yields the Vandermonde determinant in \eqref{dimensionR} 
and thus \eqref{DrSub1} follows.

A similar calculation can be done for the equivalence between \eqref{dimensionR} and \eqref{DrSub2}. For this we are going to use 
\begin{align}
\sum_{\ell=0}^{N-j} \binom{\alpha}{m_k - N + j + \ell} \binom{N-j}{\ell} = \binom{\alpha +N -j}{m_k},
\end{align}
which is a special case of \eqref{ApendixA} shown in Appendix \ref{appA}. The infinite sum given there is truncated by the vanishing of the second binomial coefficient under the sum for $\ell\geq N-j$.
This can be used to obtain the following identity of determinants,
\begin{align}
\detjk \left[\binom{\alpha +N -j}{m_k} \right] &= \detjk \left[ \sum_{\ell=0}^{N-j} \binom{\alpha}{m_k - N + j + \ell} \binom{N-j}{\ell} \right] \nonumber \\
&=  \detjk \left[ \binom{\alpha}{m_k - N + j } \right] .
\label{DetBinomProp}
\end{align}
In the second step we have again used the invariance of the determinant under the addition of multiples of rows (or columns), to remove the summands with $\ell>0$. The right-hand side is the determinant in \eqref{DrSub2}, that we will now reduce again to the Vandermonde determinant in \eqref{dimensionR}.
Using the definition of the binomial coefficient we obtain for the  left-hand side from the definition \eqref{binomDef} 
\begin{align}
\detjk \left[ \binom{\alpha +N -j}{m_k}\right] &= \detjk \left[ \frac{ \Gamma (\alpha+N-j+1) }{\Gamma\left(\alpha+N-j-m_k+1 \right) m_k!} \right] \nonumber \\
&= \detjk \left[ \frac{1}{\Gamma\left(\alpha+N-j-m_k+1 \right)}\right] \prod_{j=1}^N  \frac{\Gamma (\alpha+N-j+1)}{m_j! }.
\end{align}
Similar to the argument leading to \eqref{drcorrola1}, we can rewrite 
\begin{align}
\frac{1}{\Gamma\left(\alpha+N-j-m_k+1 \right)}=\frac{(\alpha+n-j-m_k+1)\cdots(\alpha+N-m_k-1)}{\Gamma(\alpha+N-m_k)}=\frac{(-m_k)^{j-1}+\mathcal O (m_k^{j-2})}{\Gamma(\alpha+N-m_k)},
\end{align}
to obtain
\begin{align}
 \detjk \left[ \frac{1}{\Gamma\left(\alpha+N-j-m_k+1 \right)}\right] =  \frac{
\detjk \left[ (-m_k)^{j-1} \right] 
}{\prod_{k=1}^N\Gamma \left(\alpha + N - m_k \right)}.
\end{align}
The determinant on the right-hand side is again the Vandermonde determinant  $\Delta_N \left( m_1, \dots , m_N \right)$, and with the same argument as before,  collecting all factors, we 
to get the claimed \eqref{DrSub2}.
\end{proof}

Before we go on with the character expansion, let us introduce the well-known Cauchy--Binet identity, see \cite{Mehta2}, which can be seen as a discrete version of the Andr\'eief integral identity \eqref{Andreief}.
%%%%%%%%%%%%%%%%%%%%%%%%%%%%%%%%%%%
\begin{lemma}[Cauchy--Binet Identity]
\label{CB}
Let $a_{j,m}, b_{j,m} \in \Cset$ be given complex numbers for $j=1,\dots,N$, and $m \in \Zset$. Then, it follows
\begin{align}
\sum_{m_1 > \cdots > m_N} \detij \left[ a_{j,m_k} \right] \detij \left[ b_{j,m_k} \right] = \detjk \left[ \sum_{m=-\infty}^\infty a_{jm} b_{km} \right] ,
\end{align}
under the condition that the sums appearing 
converge absolutely, for all $j,k=1,\dots,N$.
\end{lemma}
In many cases, Cauchy--Binet is the name of the lemma in the case of finite sums, whereas \cite{Mehta2} presents both finite and infinite sums as needed here.

As a last ingredient, we state the expansion of an invariant (generating) function in terms of characters. It in based on \cite{Baha}, to where we refer for a proof and a more detailed discussion.
\begin{proposition}
\label{Theorem_char_exp}
Let the Laurent series or respectively the generating functional
\begin{align}
L(t)=\sum_{n=-\infty}^\infty a_n t^n
\label{Lau}
\end{align}
converge absolutely on $\abs t =1$. For any unitary matrix $U \in U(N)$, with eigenvalues $t_j$ for $j=1,\dots,N$, we have 
\begin{align}
\prod_{j=1}^N L(t_j) = \sum_{n_1 \geq \cdots \geq n_N} \detjk \left[a_{n_j+i-j} \right] (\det\left[U \right])^{n_N} \chi_{(n_1-n_N, n_2-n_N,\dots, n_N-n_N)}\left( U \right).
\label{chex-}
\end{align}
In particular, if $a_n=0$ for $n<0$, one has
\begin{align}
\prod_{j=1}^N L(t_j) = \sum_{n_1 \geq \cdots \geq n_N \geq 0} \detjk \left[a_{n_j+i-j} \right] \chi_{(n_1,\dots,n_N)}\left(U \right).
\label{chex0}
\end{align}
\end{proposition}
Below we will only need the second part of the proposition of a Taylor series. In addition we will then have convergence on the closure of the unity disc.
 
Note that 
in Proposition~\ref{Theorem_char_exp} we can replace $U \rightarrow A U$, with an arbitrary invertible $N \times N$ matrix $A \in \GL (N,\mathbb{C})$, if $L(\mu_j)$ in \eqref{Lau} still converges absolutely in term of the eigenvalues $\mu_j$ of $AU$,
\begin{align}
\prod_{j=1}^N L(\mu_j) = \sum_{n_1 \geq \cdots \geq n_N} \detjk \left[a_{n_j+i-j} \right] \det\left[AU \right]^{n_N} \chi_{(n_1-n_N, n_2-n_N,\cdots, 0)}\left( A U \right).
\label{ce2}
\end{align}

%%%%%%%%%%%%%%%%%%%%%%%%%%%%%%%%%%%%%%%%%%%%
\section{Proof of Main Results}\label{proofs}

In this section we will present the proofs of the main results and some consistency checks in the order of appearance in Section \ref{intro}, that is Theorem~\ref{Thm-IS} followed by Corollary~\ref{Cor-TW-IS}, Theorem~\ref{Thm-GFH} and finally Corollary~\ref{cor:JIS}. In this last subsection we also present Corollary~\ref{cor:JFH}, the corresponding statement for the generalised Fisher-Hartwig integral.

\subsection{Proof of Theorem~\ref{Thm-IS}}

We begin with the character expansion of the generating function, first without an external matrix. 
Let $(t_1,\dots,t_N)$ be the eigenvalues of $U$, then we get
\begin{align}
\det [1+ U]^\alpha =\prod_{j=1}^N (1+t_j)^\alpha ,
\label{ut}
\end{align}
and similarly for $(1+U^\dagger)^\alpha$.

In order to get the needed form we use the binomial series
\begin{align}
(1+t)^\alpha =\sum_{m=0}^\infty \binom{\alpha}{m}t^m ,
\label{biS}
\end{align}
see e.g. \cite{Abramowitz}[15.1.1 and 15.1.8], which only has non negative coefficients. It converges absolutely for $|t|<1$ and for $|t|=1$ if $\Re (\alpha ) > 0$, with the binomial coefficient defined in \eqref{binomDef0}. In the case that $\Re (\alpha ) > -1$, it still converges, except for $t=-1$, and the sum terminates if $\alpha$ is a non negative integer. Since we need absolute convergence, we restrict ourselves to the case $\Re (\alpha) > 0 $ first, and get the general solution by analytic continuation. Therefore, with \eqref{biS} we get
\begin{align}
\det [1+ U]^\alpha = \prod_{j=1}^N 
\sum_{\ell_j=0}^\infty \binom{\alpha}{\ell_j}t_j^{\ell_j}.
\end{align}
Therefore, we can apply Proposition~\ref{Theorem_char_exp} (for $\Re (\alpha) > 0$) to get the following character expansion
\begin{align}
\det [1+ U]^\alpha =\sum_r\detjk \left[ \binom{\alpha}{n_k+j-k}\right] \chi_{(n_1,...,n_N)}(U), 
\label{exp}
\end{align}
where here and in the following we will abbreviate by $\sum_r:= \sum_{n_1\geq \cdots \geq n_N \geq 0 } $ the sum over irreducible representations. 
This series is absolutely convergent, due to the absolute convergence of the binomial series for $\Re ( \alpha ) >0$. When switching to $U^\dag$ and multiplying with a matrix $D\in\GL(N,\mathbb{C})$, 
\begin{align}
\det [1+ U^\dagger D]^\alpha =\sum_{n_1\geq \cdots \geq n_N \geq 0 } \detjk \left[ \binom{\alpha}{n_k+j-k}\right] \chi_{r}(U^\dagger D),
\label{det+char}
\end{align}
with the restriction that the eigenvalues $\mu_j$ of $U^\dag D$ satisfy $|\mu_j|< 1$, $j=1,\ldots,N$, or equivalently $||D||<1$ because $U$ is unitary and the norm sub-multiplicative. This is to avoid poles when continuing in $\alpha$. 
In contrast, for the exponential series $\exp (x)=\sum_{n=0}^\infty \frac{x^n}{n!} $ Proposition~\ref{Theorem_char_exp} reads 
\begin{align}
\exp \left( \Trace [AU] \right) =\sum_r \detjk \left[ \frac{1}{\Gamma(n_k+j-k+1)} \right] \chi_r(AU),
\label{expchar}
\end{align}
without restriction on the eigenvalues of $AU$ for convergence.
We can now prove \eqref{1IS}. Inserting the expansions \eqref{det+char} and \eqref{expchar} inside the integral we have
\begin{align}
  Z_\alpha^{\rm IS}\left( A,D \right) &= \int dU_N \sum_{r,r'} \detjk \left[ \frac{1}{\Gamma(n_k+j-k+1)} \right] \detjk \left[ \binom{\alpha}{n'_k+j-k} \right] 
  \chi_r(AU)\chi_{r'}(U^\dagger D)\nonumber\\
&= \sum_{r} \detjk \left[ \frac{1}{\Gamma (n_k+j-k+1)} \right] \detjk \left[ \binom{\alpha}{n_k+j-k} \right] \frac{\chi_r(AD)}{d_r},
\end{align}
due to the orthogonality \eqref{MatOg}. This expression  only depends of the eigenvalues of the product $AD$ for which there are no restrictions, as $A$ is unrestricted on the compact domain of integration of $U(N)$. 
Thus we can always rescale $A$ such that $D$ satisfies $||D||<1$.

Before we can use the Cauchy--Binet identity Lemma~\ref{CB}, to get rid of the sum over representations, we have to cancel the dimension $d_r$ in the numerator. By using Corollary~\ref{CorolaDimR}, \eqref{dimensionR} and \eqref{DrSub1}, we obtain
\begin{align}
Z_\alpha^{\rm IS}\left( A,D\right) &= \sum_{r} \left( \prod_{j=1}^N \frac{(N-j)!}{(n_j+N-j)!}\right) \detjk \left[ \binom{\alpha}{n_k+j-k} \right] \chi_r(AD)\nonumber\\ 
&= \sum_{m_1 > \cdots > m_N \geq 0} \left( \prod_{j=1}^{N} (N-j)!\right) \detjk \left[ \binom{\alpha}{m_k+j-N} \frac{1}{m_k!} \right] \frac{\detjk \left[ \mu_j^{ 2m_k} \right]}{\Delta_N(\mu^2)} \nonumber \\
&= \detjk \left[\sum_{m=0}^\infty \binom{\alpha}{m+j-N} \frac{\mu_k^{ 2m} }{m!} \right] \frac{ \prod_{j=1}^{N} (N-j)!}{\Delta_N(\mu^2)}.
\label{Test111}
\end{align}
In the second step we have used the Weyl character formula \eqref{WCF}, in terms of the eigenvalues $\mu_j^2$, of $AD$ for $j=1,\ldots,N$,  relabelled $m_j=n_j+N-j$, and multiplied the factors $1/m_j!$ inside the determinant. In the last step we applied Lemma~\ref{CB}.

Finally, we have to show that the sum inside the determinant is proportional to Kummer's confluent hypergeometric function \eqref{1F1DEF}.
Since $\binom{\alpha}{m+j-N} =0$ if $m+j-N < 0$, the sum actually starts at $m=N-j$ for $j=1,\dots,N$, and by shifting the summation index back to $0$, we get
\begin{align}
  \sum_{m=N-j}^\infty \binom{\alpha}{m+j-N} \frac{1}{m!} \mu_k^{2m} &=  \sum_{m=0}^\infty \binom{\alpha}{m} \frac{1}{(m+N-j)!} \mu_k^{2(m +N -j)}\nonumber\\
&= \sum_{m=0}^\infty \binom{m-\alpha-1}{m} \frac{(-1)^m}{(m+N-j)!}  \mu_k^{ 2(m +N -j)} \nonumber\\
&=  \sum_{m=0}^\infty \frac{\Gamma \left( m-\alpha \right)(-1)^m }{ \Gamma \left(-\alpha \right) m! (m+N-j)!} \mu_k^{ 2(m +N -j)}\nonumber\\ 
&= {{}_1F_1} (-\alpha ; N-j+1 ; - \mu_k^2) \frac{\mu_k^{2(N-j)}}{(N-j)!}.
\label{1F1power}
\end{align}
In the second step we used \eqref{bc-}. Then, we wrote out the binomial coefficient with Gamma-functions from \eqref{binomDef}, and identified the hypergeometric function from \eqref{1F1DEF}. Taking the remaining factorial out of the determinant we arrive at \eqref{1IS}.

Below we will need a second, equivalent representation, where the first argument of ${}_1F_1$ gets shifted, $-\alpha\to-\alpha-j+1$. Without going much into detail, this can be seen as follows. A similar strategy will be applied in other proofs below. We can also rewrite the determinant in the first line of \eqref{Test111}, using Corollary~\ref{CorolaDimR}, \eqref{DrSub2} and \eqref{DrSub1}:
\begin{equation}
    \detjk \left[ \binom{\alpha}{n_k+j-k}\right] = \detjk \left[ \frac{1}{\Gamma (n_k+j-k+1)}\right] \prod_{j=1}^N \frac{\Gamma (\alpha +N-j+1)}{\Gamma(\alpha - n_j + j)}.
    \label{Sub2-Sub1}
\end{equation}
Using the same steps, we arrive at
\begin{align}
Z_\alpha^{\rm IS}\left( A,D\right) 
&= \detjk \left[\sum_{m=0}^\infty \frac{1}{\Gamma(\alpha-m+N)\Gamma(m-N+j+1)}\frac{\mu_k^{ 2m} }{m!} \right] \frac{ \prod_{j=1}^{N} (N-j)!\Gamma(\alpha+j)}{\Delta_N(\mu^2)},
\end{align}
where in the last factor in the product we have shifted the index $N-j\to j-1$. Pulling all factors inside the determinant, we have to derive the claimed hypergeometric function inside. First, also here the summation 
actually starts at $m=N-j$, and relabelling we obtain
\begin{align}
\sum_{m=0}^\infty \frac{\Gamma(\alpha+j)\Gamma(N-j+1)}{\Gamma(\alpha-m+1)\Gamma(m+n-j+1)} \frac{\mu_k^{2(m +N -j)}}{m!}
&= {{}_1F_1} (-\alpha -j+1; N-j+1 ; - \mu_k^2).
\end{align}
The last step follows from the definition \eqref{1F1DEF} and the identity
\begin{align}
\frac{\Gamma(\alpha+j)}{\Gamma(\alpha+j-m)}=(-1)^m\frac{\Gamma(-\alpha-j+1+m)}{\Gamma(-\alpha-j+1)}, 
\label{GammaId}
\end{align}
which can easily be shown  along the lines of  the proof of \eqref{bc-}. Consequently, we also have that in Theorem~\ref{Thm-IS}, \eqref{1IS} can be written as 
\begin{eqnarray}
Z_\alpha^{\rm IS}( A, D ) 
&=& 
\frac{ \detjk\left[ _1F_1 (-\alpha-j+1 ; N-j+1 ; - \mu^2_k)\, \mu_k^{2(N-j)} \right] 
}{\Delta_N(\mu^2)}.
\label{1ISalt}
\end{eqnarray}
We expect, that the equivalence to \eqref{1IS} could also be shown using recurrence relations for ${}_1F_1$, together with the invariance of the determinant under row and column operations.

Let us discuss now the analytic continuation to $\alpha\in\mathbb{C}$ arbitrary. On the left-hand side in \eqref{1IS} the integral always exists for $||D||<1$, as there are no poles then and the integration domain is compact. On the right-hand side, for both equivalent forms \eqref{1IS} and \eqref{1ISalt}  the hypergeometric function \eqref{1F1DEF} is analytic and thus can be continued, due to $b=N-j+1\geq 1$ in the definition \eqref{1F1DEF}.\\

%%%%%%%%%%%%%

The proof of the second, two-fold integral \eqref{2IS} is not more difficult now. We simply have to use the orthogonality \eqref{MatOg} twice and cancel two dimensions $d_r$, by applying Corollary~\ref{CorolaDimR} twice. Let us only present the main steps here. Assuming $||C||,||D||<1$ and $\Re(\alpha)>0$, 
we can insert the expansions \eqref{det+char} and \eqref{expchar} into \eqref{2IS}: 
\begin{align}
  Z_\alpha^{\rm IS}\left( A,B,C,D \right) &= \int dU_N \int dV_N\sum_{r,r'} \detjk \left[ \frac{1}{\Gamma(n_k+j-k+1)} \right] \detjk \left[ \binom{\alpha}{n'_k+j-k} \right] \nonumber\\
  &\quad\times
  \chi_r(AUBV)\chi_{r'}(V^\dag CU^\dagger D)\nonumber\\
&= \int dU_N \sum_{r} \detjk \left[ \frac{1}{\Gamma (n_k+j-k+1)} \right] \detjk \left[ \binom{\alpha}{n_k+j-k} \right] \frac{\chi_r(AUBCU^\dag D)}{d_r},\nonumber\\
&=\sum_{r} \detjk \left[ \frac{1}{\Gamma (n_k+j-k+1)} \right] \detjk \left[ \binom{\alpha}{n_k+j-k} \right] \frac{\chi_r(AD) \chi_r(BC)}{d_r^2},
\label{2ISstep1}
\end{align}
where 
\begin{align}
  \int dU_N \chi_r(AUB CU^\dagger D ) = \frac{1}{d_r}\chi_r (AD) \chi_r (BC) 
  \label{1char}
\end{align}
follows directly from \eqref{shurOg}. Again, we initially have that the eigenvalues of 
$V^\dag CU^\dagger D$ have to have absolute values smaller than unity, or equivalently $||C||,||D||<1$. Applying the Weyl character formula \eqref{WCF}, and both identities from Corollary~\ref{CorolaDimR}, we arrive at 
\begin{align}
  Z_\alpha^{\rm IS}\left( A,B,C,D \right) &=\sum_r \left( \prod_{j=1}^{N} \frac{(N-j)!^2\Gamma(\alpha+N-j+1)}{(n_j+N-j)!^2\Gamma(\alpha-n_j+j)}\right) \frac{\detjk \left[ \mu_j^{ 2(n_k+N-k)} \right]\detjk \left[ \nu_j^{ 2(n_k+N-k)} \right]}{\Delta_N(\mu^2)\Delta_N(\nu^2)} 
\nonumber\\  
  &=\frac{\prod_{j=1}^N(N-j)!^2\Gamma(\alpha+N-j+1)}{\Delta_N(\mu^2)\Delta_N(\nu^2)}
  \detjk\left[\sum_{m=0}^\infty\frac{(\mu_j\nu_k)^{2m}}{\Gamma(\alpha+N-m)m!^2}\right],
  \label{pre2IS}
\end{align}  
where $\mu_j^2$ are the eigenvalues of $AD$, $\nu_j^2$ the eigenvalues of $BC$, for $j=1,\ldots,N$, and the Cauchy--Binet identity has been applied. 
Because $A$ and $B$ are unrestricted on the compact domains of integration for $U$ and $V$,  there are no restrictions on the absolute values of the eigenvalues. 
Thus we can rescale $C$ and $D$ to achieve $||C||,||D||<1$. It remains to show that the matrix inside the determinant again yields ${}_1F_1$. Applying \eqref{GammaId} with $j\to N$, we can rewrite the sum inside the determinant in \eqref{pre2IS} as 
\begin{align}
\frac{1}{\Gamma(\alpha+N)}\sum_{m=0}^\infty \frac{\Gamma(\alpha+N)\Gamma(1)}{\Gamma(\alpha+N-m)\Gamma(m+1)}\frac{z^m}{m!}
&=\frac{1}{\Gamma(\alpha+N)}\sum_{m=0}^\infty \frac{\Gamma(-\alpha-N+m+1)\Gamma(1)}{\Gamma(-\alpha-N+1)\Gamma(m+1)}\frac{(-z)^m}{m!}
\nonumber\\
&=\frac{1}{\Gamma(\alpha+N)}{}_1F_1(-\alpha-N+1;1;-z).
\end{align}
Collecting all the factors we thus arrive at \eqref{2IS}. The continuation to general $\alpha\in\mathbb{C}$ goes as before as due to $||C||,||D||<1$ there are no poles in the integrand, and in the hypergeometric function \eqref{1F1DEF} we have $b=1$ on the right-hand side.

%%%%%%%%%%%%%%%%%%%%%%%%%%%%%%%%%%%%%%%%%%%%%%%%%%%%%%%%%%%
\subsection{Proof of Corollary \ref{Cor-TW-IS} and consistency checks}

In this subsection we will present two proofs of Corollary~\ref{Cor-TW-IS}. The simplest and shortest way is to use the Andr\'ei\'ef integral Lemma~\ref{LemmaAii} from Appendix \ref{appB}. It is shown there how to derive the known, standard Fisher-Hartwig integral, and we will follow this derivation closely. The second way is to consider the degenerate limit $D=\pm\mathbf{1}_N$ and $A=b\mathbf{1}_N$ in \eqref{1IS} using l'H\^opital, which serves as a consistency check as well. In the same manner we then  check the reduction of \eqref{2IS} to \eqref{1IS}, by setting $B=C=\mathbf{1}_N$ and taking the degenerate limit in the same way.

\begin{proof}
Following the diagonalisation of the unitary matrix $U\in U(N)$ as in Appendix \ref{appB}, we can write \eqref{1IS'} in terms of the eigenvalues $t_j=e^{i\theta_j}$, $j=1,\ldots,N$, of $U$ as
\begin{align}
Z_{\alpha,\pm}^{\rm IS}(b)&= 
\frac{1}{(2\pi)^N N!}\prod_{l=1}^N\int_0^{2\pi}d\theta_l \left(1\pm t_l^{-1}\right)^\alpha\exp[bt_l] 
\detjk\left[t_j^{k-1}\right]\detjk\left[t_j^{-k+1}\right] \nonumber\\
&=\frac{1}{(2\pi)^N}\detjk\left[ I_{\alpha,\pm,b}^{\rm IS}(j-k)\right].
\end{align}
Here, it is clear that we have to restrict to $\Re(\alpha)>-1$ to avoid poles, and we have defined
\begin{align}
I_{\alpha,\pm,b}^{\rm IS}(m)&:=\int_0^{2\pi}d\theta \left(1\pm e^{-i\theta}\right)^\alpha\exp[be^{i\theta}]\exp[im\theta]
=\int_0^{2\pi}d\theta\sum_{n,l=0}^\infty\binom{\alpha}{n}(\pm t)^{-n}\frac{(bt)^l}{l!}t^m
\nonumber\\
&=(\pm1)^m\sum_{l=0}^\infty \binom{\alpha}{m+l}\frac{(\pm b)^l}{l!}\nonumber\\
&= (\pm1)^m\binom{\alpha}{m}{}_1F_1 (-\alpha+m;m+1;\mp b).
\label{IISfinal}
\end{align}
 In the first step above we used the expansion \eqref{biS}, with $t=e^{i\theta}$, and then applied Cauchy's theorem for the integral over the unit circle as in Appendix \ref{appB}. In the last step we utilise the identity \eqref{GammaId} and the definition \eqref{1F1DEF}.  Inserting this expression with $m=j-k$ above leads to \eqref{1IS'}, after taking out and cancelling the signs $(-1)^j$ and $(-1)^k$ from rows respectively columns. Note that from the second line in \eqref{IISfinal} is is clear, that also for $j-k<0$ negative, the matrix elements $I_{\alpha,\pm,b}^{\rm IS}(j-k)$ are non-vanishing. While the right-hand side could be continued to arbitrary values of $\alpha\in\mathbb{C}$, due to the analyticity of Kummer's hypergeometric function at least for $b>0$, the integral on the left-hand side ceases to exist for $\Re(\alpha)\leq -1$.
 
The second proof, and at the same times consistency check, is obtained by taking the degenerate limit $D=\pm\mathbf{1}_N$ and $A=b\mathbf{1}_N$ in \eqref{1IS}, or simply $\mu_k^2\to \pm b$ for all $k=1,\ldots,N$. Using l'H\^opital, we have 
 \begin{align}
 Z_{\alpha,\pm}^{\rm IS}(b)&= 
\lim_{\mu^2_{j=1,\ldots,N}\to\pm b}Z_\alpha^{\rm IS}( A, D ) \nonumber\\
&=\left. 
\frac{ \prod_{l=1}^N\frac{d^{N-l}}{d(\mu_l^2)^{N-l}}\detjk\left[ _1F_1 (-\alpha-j+1 ; N-j+1 ; - \mu^2_k)\, \mu_k^{2(N-j)} \right] 
}{\prod_{l=1}^N\frac{d^{N-l}}{d(\mu_l^2)^{N-l}}\Delta_N(\mu^2)}\right|_{\mu^2_{j=1,\ldots,N}=\pm b}.
\label{lim1IS}
\end{align}
For the derivatives we may write  
\begin{equation}
\frac{d^n}{dt^n} t^m=\frac{m!}{\Gamma(m-n+1)}t^{m-n},
\label{dntm}
\end{equation} 
which vanishes for $m<n$ as it should. The derivatives can be taken into the rows of both determinants, and we obtain for the Vandermonde determinant
\begin{align}
\detjk \left[\left. \frac{d^{N-j}}{d (\mu_j^2)^{N-j}} \mu_j^{2(N-k)} \right|_{\mu_j^2=\pm b} \right] = \prod_{j=1}^N \left( \left. \frac{d^{N-j}}{d (\mu_j^2)^{N-j}} \mu_j^{2(N-j)} \right|_{\mu_j^2=\pm b} \right) = \prod_{j=1}^N  (N-j)! 
\label{dr3deno}
\end{align} 
Because of \eqref{dntm} the differentiation leads to an upper triangular matrix inside the determinant, and thus to  the product of the diagonal elements.\footnote{In the same way \eqref{dimensionR} can be obtained from \eqref{WCF} by taking the degenerate limit $U\to\mathbf{1}^N$, $t_j\to1$ $\forall j$, to the identity.}
For the differentiation in the numerator we go back to \eqref{1F1power}, with shifted index $m\to m+N-j$: 
\begin{align}
{{}_1F_1} (-\alpha ; N-j+1 ; - \mu_k^2) \mu_k^{2(N-j)}
=  \sum_{m=0}^\infty \frac{\Gamma(-\alpha+m-N+j)\Gamma(N-j+1)(-1)^{N-j} }{ \Gamma (-\alpha) \Gamma(m-N+j+1)} \frac{(-\mu_k^2)^{ m}}{m!}. 
\end{align}
Thus after differentiation using \eqref{dntm}, and shifting again indices to $m-N+k$,  we obtain
\begin{align}
&\frac{d^{N-k}}{d (\mu_k^2)^{N-k}} {{}_1F_1} (-\alpha ; N-j+1 ; - \mu_k^2) \mu_k^{2(N-j)}|_{\mu_k^2=\pm b}
\nonumber\\
&= \sum_{m=0}^\infty \frac{\Gamma(-\alpha+m-N+j)\Gamma(N-j+1)(-1)^{N-j+m} (\pm b)^{ m-N+k}}{ \Gamma (-\alpha) \Gamma(m-N+j+1)\Gamma(m-N+k+1)}
\nonumber\\
&= (N-j)!\sum_{m=0}^\infty \frac{\Gamma(-\alpha+m-k+j)(-1)^{j-k} (\mp b)^{ m}}{ \Gamma (-\alpha) \Gamma(m+j-k+1)\Gamma(m+1)}
\nonumber\\
&=(N-j)! \frac{(-1)^{j-k}\Gamma(-\alpha+j-k)}{\Gamma(-\alpha)(j-k)!}{{}_1F_1} (-\alpha +j-k; j-k+1 ; \mp b^2). 
\end{align}
While the first factor in the last line cancels with the one from the Vandermonde determinant in the denominator, when taken out of the determinant, we can use again the identity \eqref{bc-}, to obtain the binomial times the desired hypergeometric function inside the determinant in \eqref{1IS'}. 
 \end{proof}

As a further and final consistency check of Theorem~\ref{Thm-IS}, we would like to prove that \eqref{2IS} reduces to \eqref{1IS}, when setting $B=C=\mathbf{1}_N$. We can then use the invariance of the Haar measure $dU_N$ under the transformation $U\to UV^\dag$, in order to absorb the $V$-dependence of the integrand completely. Upon integrating out $V$ now, which gives unity as $\int dV_N=1$, we should arrive at \eqref{1IS}. In order to check this, we have to take the degenerate limit for the eigenvalues of $BC$,  $\nu^2_k\to1$, for $k=1,\ldots, N$, as in \eqref{lim1IS} for the one-matrix integral. Because we have already computed the derivative of the Vandermonde determinant, we can be brief. We only need the following formula for the $n$-fold derivative of the hypergeometric function ${}_1F_1$, see 
\cite{NIST}[13.3.16]:
\begin{align}
\frac{d^n}{dz^n}{}_1F_1(a;b;z)= \frac{\Gamma(a+n)\Gamma(b)}{\Gamma(a)\Gamma(b+n)}{}_1F_1(a+n;b+n;z).
\label{dn1F1}
\end{align}
Pulling the derivatives from the rule of l'H\^opital inside the determinants as before, we thus obtain
\begin{align}
&\left. 
\frac{ 
\detjk\left[ \frac{d^{N-k}}{d(\nu_k^2)^{N-k}}
{}_1F_1 (-\alpha-N+1 ; 1 ; - \mu^2_j\nu_k^2)\right]
}{\prod_{l=1}^N\frac{d^{N-l}}{d(\nu_l^2)^{N-l}}\Delta_N(\nu^2)}\right|_{\nu^2_{j=1,\ldots,N}=1}\nonumber\\
&=\frac{\detjk\left[\frac{\Gamma(-\alpha-k+1)\Gamma(1)(-\mu_j^2)^{(N-k)}}{\Gamma(-\alpha-N+1)\Gamma(1+N-k)} {}_1F_1 (-\alpha-N+1+N-k ; 1+N-k ; - \mu^2_j)\right]}{\prod_{j=1}^N(N-j)!}.
\end{align}
Using again the identity \eqref{GammaId} with $j=N$ and $m=N-k$, the resulting ratio of Gamma-functions can be taken out of the determinant. Together with the factor form the Vandermonde determinant  they cancel the pre-factor in \eqref{2IS}, that equivalently reads 
$g_N(\alpha)=\prod_{j=1}^N\frac{(N-j)!^2\Gamma(\alpha+j)}{\Gamma(\alpha+N)}$. Thus we arrive at the claimed identity, this time in the form \eqref{1ISalt} equivalent to \eqref{1IS}.

%%%%%%%%%%%%%%%%%%%%%%%%%%%%%%%%%%%%%%%%%%%%%%%%%%%%%%%%%%%
\subsection{Proof of Theorem~\ref{Thm-GFH}}

\begin{proof}
This proof is similar to the proof of Theorem~\ref{Thm-IS}, and we begin with proving \eqref{1FH}. We start again with the character expansion \eqref{det+char}, which requires that the eigenvalues of $A$ and $D$ lie inside the unit circle, $||A||,||D||<1$, and $\Re(\alpha),\Re(\beta)>0$. Inserting \eqref{det+char} into the integral \eqref{1FH}, due to the orthogonality \eqref{MatOg} we obtain
\begin{eqnarray}
Z_{\alpha,\beta}^{\rm FH}( A, D ) &=& \sum_{r ,r'}\int  dU_N \detjk \left[ \binom{\alpha}{n_k+j-k}\right] \detjk \left[ \binom{\beta}{n'_k+j-k}\right]\chi_{r}(A U)  \chi_{r'}(U^\dagger D) 
\nonumber\\
&=& \sum_{r} \detjk \left[ \binom{\alpha}{n_k+j-k}\right] \detjk \left[ \binom{\beta}{n_k+j-k}\right] \frac{\chi_{r}(A D)}{d_r} .
\label{proof1FHstart}
\end{eqnarray}
We can now use Corollary~\ref{CorolaDimR}, namely \eqref{DrSub2} to cancel one of the determinants with the $d_r$ in the denominator, as in the previous proof. However, to keep the symmetry under the  interchange of $\alpha$ and $\beta$ manifest, we also rewrite the other determinant as in \eqref{Sub2-Sub1}, with $\alpha\to\beta$.
Therefore, by applying \eqref{DrSub2} and \eqref{Sub2-Sub1} we have
\begin{align}
Z_{\alpha,\beta}^{\rm FH}(A,D)&= \sum_{r}  \detjk \left[ \frac{1}{\Gamma (n_k+j-k+1)}\right] \prod_{j=1}^N \frac{\Gamma (\beta +N-j+1) \Gamma (\alpha +N-j+1) (j-1)!}{\Gamma(\beta - n_j + j)\Gamma(\alpha - n_j + j)\Gamma(n_j+N-j+1)} 
\nonumber\\
&\quad \times
\chi_{r}(A D) \nonumber \\
&= \sum_{m_1> \cdots > m_N\geq 0}  \detjk \left[  \frac{\Gamma (\beta +j) \Gamma (\alpha +j) (N-j)!}{\Gamma(\beta - m_j + N)\Gamma(\alpha - m_j + N)m_j! \Gamma (m_k+j-N+1)} \right] \nonumber\\
&\quad\times
 \frac{\detjk[\mu_j^{2m_k}]}{\Delta_N(\mu^2)}\nonumber\\
 &= \detjk \left[ \sum_{m=0}^{\infty}  \frac{\Gamma (\beta +j) \Gamma (\alpha +j) (N-j)!}{\Gamma(\beta - m + N)\Gamma(\alpha - m + N)m! \Gamma (m+j-N+1)} \mu_k^{2m} \right] \frac{1}{\Delta_N(\mu^2)}.
\end{align}
In the second equality we changed the summation to $m_j=n_j+N-j$, reordered the product, letting $j\to N-j$, and applied Weyl's character formula \eqref{WCF} in terms of the eigenvalues $\mu^2_j$ of $AD$, with $|\mu_j|<1$  for $j=1,\dots,N$. The Cauchy--Binet identity Lemma~\ref{CB} leads to the final line.

To obtain \eqref{1FH} we only need to show that the sum inside the determinant is Gau\ss'  hypergeometric function \eqref{2F1Def}. First, since $\frac{1}{\Gamma(m+j-N+1)}=0$ if $m+j-N<0$, we shift the summation index back to start with 0. Second, we use \eqref{GammaId} for the $\alpha$- and $\beta$-dependent part with $N=j$. Put together we have
\begin{eqnarray}
& &\sum_{m=N-j}^{\infty}  \frac{\Gamma (\beta +j) \Gamma (\alpha +j) (N-j)!}{\Gamma(\beta - m + N)\Gamma(\alpha - m + N)m! \Gamma (m+j-N+1)} \mu_k^{2m} \nonumber \\
&=& \sum_{m=0}^{\infty}\frac{\Gamma (\beta +j) \Gamma (\alpha +j) (N-j)!}{\Gamma(\beta - m + j)\Gamma(\alpha - m + j) \Gamma(m+N-j+1)! } \frac{\mu_k^{2m+2(N-j)}}{m!} \nonumber \\
&=& \sum_{m=0}^{\infty}  \frac{\Gamma (-\beta  -j+1+m) \Gamma (-\alpha -j+1+m) (N-j)!}{\Gamma(\beta -j+1)\Gamma(\alpha -j+1) \Gamma(m+N-j+1)! } \frac{\mu_k^{2m+2(N-j)}}{m!} \nonumber \\ 
&=& {{}_2F_1}(-\alpha-j+1,-\beta -j +1; N-j+1; \mu^2_k)\mu_k^{2(N-j)}.
\end{eqnarray}
The last step directly follows from 
\eqref{2F1Def} and gives the claim. It is manifestly symmetric in $\alpha$ and $\beta$.

Had we only expanded one of the determinants in \eqref{proof1FHstart} using Corollary~\ref{CorolaDimR}, we would arrive at an equivalent result, that is not explicitly symmetric in $\alpha$ and $\beta$. Leaving the details to the reader, we only quote the final answer for this equivalent form of \eqref{1FH} in Theorem~\ref{Thm-GFH}:
\begin{eqnarray}
Z_{\alpha,\beta}^{\rm FH}( A, D ) 
&=& 
\frac{ \detjk\left[ _2F_1 (-\alpha-j+1;-\beta ; N-j+1 ; \mu_k^2)\, \mu_k^{2(N-j)} \right]
}{\Delta_N(\mu^2)}.
\label{1FH'}
\end{eqnarray}
Yet another form exists with $\alpha\leftrightarrow\beta$ exchanged and the first and second argument of ${}_2F_1$ exchanged. 
Once again we expect that the equality between \eqref{1FH} and \eqref{1FH'}  (or the third form) can shown independently, using recurrence relations for Gau{\ss}'s hypergeometric function, together with row and column operations under the determinant.
The continuation in $\alpha$ and $\beta$ goes as follows. In the integrals on the right-hand side there are no poles for $||A||,||D||<1$, and thus the integral is well-defined for $\alpha,\beta\in\mathbb{C}$. Without restriction, however, we have to impose at least $\Re(\alpha+\beta)>-1$ as in the original Fisher-Hartwig integral \eqref{FHdef}. 
On the right-hand side, due to $c=N-j+1\geq1$ in the definition \eqref{2F1Def} in all equivalent forms \eqref{1FH} and \eqref{1FH'} we can continue in both arguments in $\alpha$ and $\beta$.\\ 

It remains to prove the second statement in the theorem, the two-matrix integral \eqref{2FH}. 
Since $AUB $, $C U^\dagger D \in \GL(N,\mathbb{C})$, assuming that the eigenvalues of $A,B,C$ and $D$  lie inside the unit circle,
we can proceed as in the previous case \eqref{2ISstep1} when $\Re(\alpha),\Re(\beta)>0$. Using \eqref{1char} we obtain from the first and second group integral
\begin{align}
  Z_{\alpha, \beta }^{FH} \left( A, B , C , D \right) &=\int dU_N \sum_{r}  \frac{\chi_r(AUBC U^\dagger D )}{d_r}  
  \detjk \left[ \binom{\alpha}{n_k+j-k}\right] \detjk \left[ \binom{\beta}{n_k+j-k}\right] \nonumber\\
&= \sum_{r} \detjk \left[ \binom{\alpha}{n_k+j-k}\right] \detjk \left[ \binom{\beta}{n_k+j-k}\right] 
 \frac{\chi_r (AD) \chi_r (BC) }{d_r^2} .
\end{align}
It only depends on the eigenvalues $\mu_j^2$ of $AD$ and the eigenvalues $\nu_j^2$ of $BC$, restricted to $|\mu_j| < 1 $ and $|\nu_j| < 1 $ for all $j=1,\dots,N$. Applying  Corollary~\ref{CorolaDimR} twice,  as well as Weyl's character formula \eqref{WCF}, we get
\begin{align}
  Z_{\alpha, \beta }^{FH} \left( A, B , C , D \right) &= \sum_{r} \prod_{j=1}^N \left( \frac{(N-j)!^2 \Gamma(\alpha+ N-j+1) \Gamma(\beta+ N-j+1)}{m_j!^2 \Gamma(\alpha -m_j+N) \Gamma(\beta-m_j+N)} \right) 
\nonumber \\ 
&\quad\times 
 \frac{\detjk \left[ \mu_j^{ 2(n_k+N-k)} \right]\detjk \left[ \nu_j^{ 2(n_k+N-k)} \right]}{\Delta_N(\mu^2)\Delta_N(\nu^2)}  
 \nonumber\\
&=  \detjk \left[ \sum_{m=0}^\infty \frac{1}{m!^2 \Gamma(\alpha -m+N) \Gamma(\beta -m+N)}  \mu_i^m \nu_j^m \right] \nonumber \\ 
&\quad\times
 \frac{ \prod_{j=1}^N \left( (N-j)!^2 \Gamma(\alpha+ N-j+1) \Gamma(\beta+ N-j+1) \right) }{\Delta_N(\mu) \Delta_N(\nu)},
 \label{4.4ZFHproto}
\end{align}
after using the Cauchy--Binet Lemma~\ref{CB} again. Finally, the sum inside the determinant can also be related to  Gau{\ss}'s hypergeometric function 
\begin{align}
&\frac{1}{\Gamma(\alpha+N)\Gamma(\beta+N)}  \sum_{m=0}^\infty \frac{\Gamma(\alpha+N)\Gamma(\beta+N)}{m!^2 \Gamma(\alpha -m+N) \Gamma(\beta -m+N)}  \mu_i^m \nu_j^m  
\nonumber\\
 &=\frac{1}{\Gamma(\alpha+N)\Gamma(\beta+N)}  
  \sum_{m=0}^\infty \frac{\Gamma(-\alpha -N+m+1)\Gamma(-\beta -N+m+1)\Gamma(1)}{\Gamma(-\alpha -N+1) \Gamma(-\beta -N+1)\Gamma(m+1)} \frac{\mu_i^m \nu_j^m}{m!} \nonumber \\ 
  &=\frac{1}{\Gamma(\alpha +N) \Gamma(\beta +N)} {}_2F_1 (-\alpha -N +1, - \beta -N+1 ; 1 ; \mu_j \nu_k) .
\end{align}
In the first step we have used the identity \eqref{GammaId} for $j=N$, and then applied the definition \eqref{2F1Def}. The remaining prefactors can be pulled out of the determinant, and together with the factors in \eqref{4.4ZFHproto} yield the result given in \eqref{2FH}. 

The continuation in $\alpha$ and $\beta$ goes in the same way as before, depending on whether  the norms of $A,B,C$ and $D$ are all less than one or not.
\end{proof}
Similarly to the end of the last subsection, there is a consistency check in setting $B=C=1$ to reobtain \eqref{1FH} from \eqref{2FH}. We leave this to the reader. 

%%%%%%%%%%%%%%%%%%%%%%%%%%%%%%%%%%%%%%%%%%%%%%%%%%%%%%%%%%%
\subsection{Proof of Corollary~\ref{cor:JIS}}
%%%%%%%%%%%%%%%%%%%%%%%%%%%%%%%%%%%%%%%%%%%%%%%%%%%%%%%%%%%

In this subsection we will prove Corollary~\ref{cor:JIS}, by integrating \eqref{2IS} with $B=C=X$ over the Hermitian random matrix $X=X^\dag$, with an appropriate measure $\exp[-\Trace v(X)]$, such that the integral is well defined. A similar result is obtained when integrating \eqref{2FH} under the same conditions, and is stated in Corollary~\ref{cor:JFH} at the end. Because the proof is completely analogous we will omit it here. 

\begin{proof}
Consider the matrix integral as given in  \eqref{JSPfaff} in Corollary~\ref{cor:JIS}. Diagonalising the random matrix $X=WxW^\dag$, with $W\in U(N)$ and $x=\mbox{diag}(x_1,\ldots,x_N)$, gives a known Jacobian, and we have 
\begin{align}
\int dH_N = c_N \int dU_N \int_{\mathbb{R}^N}dx_1\cdots dx_N \Delta_N(x)^2. 
\end{align}
Given our convention $\int dU_N=1$, the constant $c_N={\pi^{\frac{N(N-1)}{2}}}/{\prod_{j=1}^N j!}$ can be determined for example by comparing the Gaussian integral of the Gaussian Unitary Ensemble GUE with $v(x)=X^2$ over independent matrix elements with the product of the norms of the Hermite polynomials, resulting from the Andr\'ei\'ef integral applied to the squared Vandermonde. 
With these prerequisites, we have 
\begin{eqnarray}
J_{\alpha}^{\rm IS}(A,D) &=& \int dH_N e^{-\Trace v(X)}Z_{\alpha}^{\rm IS}(A,X,X,D)\nonumber\\
&=& c_N \prod_{i=1}^N\int_{\mathbb{R}}dx_ie^{-v(x_i)}\Delta_N(x)^2 \int  dU_N\!\int dV_N \det \left[1+ V^\dag xU^\dagger D\right]^\alpha \exp \left[ \Trace (AUxV) \right]
\nonumber\\
&=& c_N g_N(\alpha) \prod_{i=1}^N\int_{\mathbb{R}}dx_ie^{-v(x_i)}\Delta_N(x)^2 
\frac{ \detjk\Big[ {}_1F_1 (-\alpha -N+1 ; 1; -a_j^2x_k^2)\Big] 
}{\Delta_N(a^2)\Delta_N(x^2)},
\label{2ISint1}
\end{eqnarray}
where the right and left  invariance of the Haar measures $dU_N$ and $dV_N$, respectively, has been used to absorb the diagonalising matrix $W$ in the second line. We have inserted the result \eqref{2IS}, recalling that here we call $a_k^2$ the eigenvalues of  $AD$, and $x_k^2$ the eigenvalues of $X^2$, for $k=1,\dots,N$. We can now apply the Schur-Pfaff identity \cite{Schur}:
\begin{align}
  \frac{\Delta_N (x)^2}{\Delta_N (x^2)} = \left\{
\begin{array}{ll}
\Pf \left[ \frac{x_j -x_k }{x_j + x_k } \right]_{j,k=1}^N & ,\ N \text{ even}, \\
\Pf\left[ \begin{array}{ c | c }
\frac{x_j -x_k }{x_j + x_k }  & 1 \\ \hline
-1 & 0 \\
\end{array} \right] &,\ N \text{ odd}, \\
\end{array}
\right. 
\label{SchurPfaff}
\end{align}
quoted for even and odd $N$. For the definition of the Pfaffian determinant and its properties see e.g.  \cite{Mehta2}.  Let us first look at even $N$. We can apply an integration theorem by de Bruijn \cite{deBru}, which holds for $s(x,y) = -s(y,x)$ an anti-symmetric function, and $\phi_j(x)$ a set of functions, such that all integrals exist with respect to measure $d\mu$:
\begin{align}
 &\int d\mu(x_1) \cdots d\mu(x_N) \; \Pf \left[ s(x_j , x_k) \right]_{j,k=1}^N \detjk \left[ \phi_j (x_k) \right] \nonumber \\ 
  &= N! \Pf 
  \left[ \int d\mu( x) d\mu(y) \; s(x,y) \frac{1}{2}\left( \phi_j (x) \phi_k (y) - \phi_k (x) \phi_j (y) \right)
  \right]_{j,k=1}^N .
  \label{deBruijnE}
\end{align}
Here we have already anti-symmetrised the double integral inside the determinant, as it will be needed now. Due to the Schur-Pfaff identity \eqref{SchurPfaff} we have $s(x,y)=(x-y)/(x+y)$, which has a pole at $x=-y$, a weight 
$d\mu(x)=dx e^{-v(x)}$ on $\mathbb{R}$, and $\phi_j(x_k)= {}_1F_1 (-\alpha -N+1 ; 1; -a_j^2x_k^2)$. Collecting all prefactors we thus arrive exactly at \eqref{JSPfaff} for even $N$, together with the definition \eqref{EISajak}. Because in our case $\phi_j(-x_k)=\phi_j(x_k)$ in an even function, the pole at $x=-y$ inside the integral $E_\alpha^{\rm IS}(a_j,a_k) $ cancels. 

For $N$ odd a similar formula of de Bruijn holds:
\begin{align}
 &\int d\mu(x_1) \cdots  d\mu(x_N) \; 
\Pf\left[ \begin{array}{ c | c }
s(x_j , x_k) & 1 \\ \hline
-1 & 0 \\
\end{array} \right]_{j,k=1}^N 
 \detjk \left[ \phi_j (x_k) \right] \nonumber \\ 
  &= N! 
\Pf\left[ \begin{array}{ c | c }
\int d\mu( x) d\mu(y) \; s(x,y) \frac{1}{2}\left( \phi_j (x) \phi_k (y) - \phi_k (x) \phi_j (y) \right) & 
\int d\mu(x) \phi_j(x) \\ \hline
-\int d\mu(x) \phi_k(x) & 0 \\
\end{array} \right]_{j,k=1}^N.      
  \label{deBruijnO}
\end{align}
Because of the additional last row and column, the matrix inside the Pfaffian determinants are again of even dimension. 
In the same way as before, applying this integral identity to \eqref{2ISstep1} for odd $N$, together with \eqref{SchurPfaff} we arrive at \eqref{JSPfaff} together with \eqref{FISaj} in the case of odd $N$. 
\end{proof}

In the very same way as described above we can make $B=C=X$ random in Theorem~\ref{Thm-GFH} in \eqref{2FH}, and integrate over $X$. Because the proof goes exactly along the same lines we only state the final result. 
%%%%%%%%%%%%%%
\begin{corollary}[Integrated Generalised Fisher-Hartwig Integral]\label{cor:JFH}
Let $A,D\in{\rm GL}(N,\mathbb{C})$ be fixed complex invertible $N\times N$ matrices and denote by $a_j^2$, $j=1,\ldots,N$, the eigenvalues of $AD$. 
The matrix $X=X^\dag$ is a complex Hermitian random matrix distributed according to $\exp[-\Trace v(X)]$, where $v$ decays sufficiently fast at infinity. Furthermore, let $\Re(\alpha+\beta)>-1$. Then the following integral identity holds:
\begin{eqnarray}
J_{\alpha,\beta}^{\rm FH}(A,D)
&:=& \int dH_N e^{-\Trace v(X)}\int dU_N \int dV_N
\det \Big[1+ AUXV\Big]^\alpha\det \left[1+ V^\dag XU^\dagger D\right]^\beta
\nonumber\\
&=&c_N\frac{N!  g_N(\alpha)g_N(\beta)}{\Delta_N(a^2)}  \left\{
\begin{array}{ll}
\Pf \left[E_{\alpha,\beta}^{\rm FH}(a_j,a_k) \right]_{j,k=1}^N &, N \text{ even}, \\[2mm]
\Pf\left[ \begin{array}{ c | c }
E_{\alpha,\beta}^{\rm FH}(a_j,a_k) & F_{\alpha,\beta}^{\rm FH}(a_j) \\\hline
-F_{\alpha,\beta}^{\rm FH}(a_k) & 0 \\
\end{array} \right]_{j,k=1}^N &, N \text{ odd}, \\
\end{array}
\right.
\label{JFHPFaff}
\end{eqnarray}
with
\begin{eqnarray}
 \label{EFHajak}
 E_{\alpha,\beta}^{\rm FH}(a_j,a_k) &:=& \frac12\int_{\Rset^2} dxdy \frac{x-y}{x+y} e^{-v(x)-v(y)} 
 \\ 
 &&\times \,\bigg( {}_2F_1 (1-\alpha-N,1-\beta-N;1;a^2_jx^2) {}_2F_1 (1-\alpha-N,1-\beta -N; 1 ; a^2_k y^2) \nonumber  \\ 
 &&\quad- {}_2F_1 (1-\alpha-N,1-\beta-N;1;a^2_kx^2) {}_2F_1 (1-\alpha-N,1-\beta -N; 1 ; a^2_j y^2)  \bigg),
\nonumber
\end{eqnarray}
and 
\begin{equation}
F_{\alpha,\beta}^{\rm FH}(a_j) := \int_{\Rset} dx  e^{-v(x)} {}_2F_1 (1-\alpha -N, 1- \beta -N ; 1 ; a^2_j x^2)  . 
 \label{FFHaj}
\end{equation}
\end{corollary}
It is possible that for certain choices of the weight function $e^{-v(x)}$ the integrals can be simplified. Even for the Gaussian case when $X$ belongs to the GUE, this does not seem to be an easy task.

%%%%%%%%%%%%%%%%%%%%%%%%%%%%%%%%%%%%%%%%%%%%
\section{Conclusions and Open Questions}\label{conclusio}

In this paper we have derived explicit determinantal and Pfaffian formulas for certain integrals over the unitary group with respect to Haar measure. They depend on fixed external matrices and have been named after known matrix-integrals where such matrices are absent, namely the Fisher-Hartwig and Ingham-Siegel integral. The later was initially defined by an integration over Hermitian random matrices. Here, instead we obtained unitary one- and two-matrix versions for both kinds of integrals, with explicit determinantal formulas in terms of Kummer's respectively Gau{\ss}' hypergeometric function. 
When some of the fixed, external matrices were made random and integrated over, we obtained Pfaffian determinants. 

We expect that our generalised group integrals will find applications in physics and mathematics. Most likely, the cases $A=D$ and $B=C$ Hermitian will be encountered, but we choose to present the most general form. 
Regarding possible mathematical extensions, two questions naturally arise. First, one might wonder if higher order multi-unitary group integrals could be solved, too. Because each integration results into an addition factor $1/d_r$ from the character expansion, it will be difficult to put the result into a form where the Cauchy--Binet identity can be applied. A second obvious question is about other compact group integrals, for example over the orthogonal or unitary-symplectic group. While the character expansion remains an available tool, already for simpler integrals this has turned out to be a formidable task.

\section*{Acknowledgments}
The work of Gernot Akemann was partly funded by the Deutsche Forschungsgemeinschaft (DFG) grant SFB 1283/2 2021 -- 317210226 and the work of Tim W\"urfel was partly supported by EPSRC Grant EP/V002473/1 "Random Hessians and Jacobians: theory and applications".
We thank Yan Fyodorov for useful discussions and Mario Kieburg for very detailed comments on the manuscript.

%%%%%%%%%%%%%%%%%%%%%%%%%%%%%%%%%%%%%%%%%%%%%%%%%%%%%%%%%%%%%%%%%%%%%
\begin{appendix}
%%%%%%%%%%%%%%%%%%%%%%%%%%%%%%%%%%%%%%%%%%%%%%%%%%%%%%%%%%%%%%%%%%%%%
\section{Properties of Binomial Coefficients and Gau{\ss}' Hypergeometric Function}\label{appA}

In this appendix, we will show a property of binomial coefficients given in  
Lemma~\ref{ApendixATheorem} below. It follows from an identity for Gau{\ss}' hypergeometric function. First let us define  the binomial coefficient for complex argument,
\begin{align}
  \binom{\alpha}{k}= \frac{\alpha (\alpha -1) \cdot \dots \cdot (\alpha - k + 1)}{k!} ,\quad \text{for } \alpha \in \Cset \text{ and } k \in \Nset_0 .
\label{binomDef0}
\end{align}
We can express the binomial coefficient via the Gamma-function, by using the property $\Gamma(z+1) = z\Gamma (z)$. It then holds
\begin{align}
  \binom{\alpha}{k} = \frac{\Gamma (\alpha +1)}{\Gamma (k+1) \Gamma (\alpha - k + 1)},
\label{binomDef}
\end{align}
up to singularities if $\alpha$ is a negative integer. The identity remains true in the limit when $\alpha$ becomes a negative integer, in the l'H\^opital sense. The right-hand side of \eqref{binomDef} is also defined for integer or even complex $k\in \Cset$, if $\alpha$ is not a negative integer.

%%%%%%%%%%%%%%%%%%%%%%%%%%%%%%%%%%%%%%%%%%%%%%%%
\begin{lemma}
Let $m \in \Zset$ and $\alpha , \beta \in \Cset$ with $\Re (\alpha + \beta ) > -1$, then
\begin{align}
\binom{\alpha+\beta}{\beta+m}=\sum_{k=0}^\infty \binom{\alpha}{m+k} \binom{\beta}{k}.
\label{ApendixA}
\end{align}
This statement remains true if $\alpha $ or $ \beta$ is a non negative integer and in this case the sum is finite.
\label{ApendixATheorem}
\end{lemma}
\begin{proof}
For the proof we first want to show that
\begin{align}
\binom{\alpha}{m} {}_2F_1(-\alpha+m,-\beta;m+1;z) = \sum_{k=0}^\infty \binom{\alpha}{m+k} \binom{\beta}{k}z^k,
\label{Fid}
\end{align}
with the hypergeometric function ${}_2F_1(a,b;c;z) $ given by the Gau{\ss} 
series \eqref{2F1Def}. 
It converges absolutely on $|z|=1$ if $\Re \left( c-a-b\right)=\Re(m+1-(m-\alpha)+\beta)>0$, which is the case here,  or if $a$ or $b$ is a non negative integer, since then the sum is finite. Below we will have to set $z=1$ eventually, to prove \eqref{ApendixA}.
Let us show the relation \eqref{Fid} by using 
\begin{align}
\binom{\alpha}{k}=(-1)^k\binom{k-\alpha-1}{k}.
\label{bc-}
\end{align}
It follows by inserting the definition \eqref{binomDef0} into the right hand side,
\begin{equation}
\binom{k-\alpha-1}{k}=\frac{(k-\alpha-1)(k-\alpha-2)\cdots(-\alpha-1)(-\alpha)}{k!}= (-1)^k\binom{\alpha}{k}.
\end{equation}
and taking out the minus sign.
Multiplying the hypergeometric function with \eqref{bc-} and spelling out the definition of ${}_2F_1$ from \eqref{2F1Def}, we obtain
\begin{align}
\binom{\alpha}{m} {}_2F_1(-\alpha+m,-\beta;m+1;z) 
&=(-1)^m\binom{-\alpha+m-1}{m} \sum_{k=0}^\infty \frac{\Gamma \left( m+1 \right) \Gamma \left(-\alpha+m +k \right)\Gamma \left( -\beta+k \right)}{\Gamma \left( -\alpha+m \right)\Gamma \left( -\beta \right)\Gamma \left(m+1+k \right) }\frac{z^k}{k!}\nonumber\\
&=\sum_{k=0}^\infty (-1)^m \frac{\Gamma\left(-\alpha+m+k \right)}{\Gamma\left(-\alpha \right)\Gamma \left(m+1+k \right)}\frac{ \Gamma \left( -\beta+k \right)}{\Gamma \left( -\beta \right)\Gamma \left(k+1 \right)}z^k\nonumber\\
&= \sum_{k=0}^\infty (-1)^{m+k} \binom{-\alpha+m+k-1}{m+k} (-1)^{k}\binom{-\beta+k-1}{k} z^k \nonumber \\
&= \sum_{k=0}^\infty \binom{\alpha}{m+k}\binom{\beta}{k}z^k .
\label{FeqSum3}
\end{align}
In the last step we have again used \eqref{bc-}. For $z=1$ we can use \cite[15.2.1]{NIST} 
\begin{align}
{}_2F_1(a,b;c;1)=\frac{\Gamma \left( c \right) \Gamma \left( c-a-b \right)}{\Gamma \left( c-a \right) \Gamma \left( c-b \right)} , \qquad \text{if } \Re(c-a-b)>0,
\label{2F1forz1}
\end{align}
on the left-hand side. If $a$ or $b$ are non negative integers, \eqref{2F1forz1} holds without the restriction $ \Re(c-a-b)>0$, see \cite[15.4.24]{NIST}. If $\Re \left(\alpha+\beta+1\right)>0$ or if $\alpha $ or $ \beta$ is a non negative integer, we get
\begin{align}
\binom{\alpha}{m} {}_2F_1(-\alpha+m,-\beta;m+1;1)= \frac{\Gamma \left( \alpha+\beta+1 \right)}{ \Gamma \left( \alpha-m+1 \right) \Gamma \left(\beta+m+1 \right)}.
\end{align}
This yields the left-hand side of the claim \eqref{ApendixA}. 
\end{proof}

%%%%%%%
\section{Andr\'ei\'ef Identity and Standard Fisher-Hartwig Integral}\label{appB}

In this appendix we briefly present the Fisher-Hartwig integral without external fields, and show how it is obtained using Andr\'ei\'ef's integral identity, given in 
Lemma~\ref{LemmaAii} below. The reason is twofold, we can then compare with \eqref{1FH} in the limit $A=D=\mathbf{1}_N$ as a consistency check. Second, this provides an  example how such unitary group integrals are determined when they do not depend on fixed external matrices, as mentioned in the introduction.

We first quote the integral identity named after Andr\'ei\'ef, \cite{An}.
\begin{lemma}[Andr\'eief's Integral Identity]
\label{LemmaAii}
Let $\phi_j(x)$ and $\psi_j(x)$ for $j=1,\dots,N$ be integrable functions, with respect to some measure $d\mu$, then the following formula holds
\begin{align}
\prod_{l=1}^N\int d \mu(x_l) \underset{1\leq j,k\leq N}{\det} \left[\phi_j(x_k) \right]\underset{1\leq j,k\leq N}{\det}\left[\psi_j(x_k) \right] = N!\detjk \left[ \int d \mu(x) \phi_j(x) \psi_k(x) \right].
\label{Andreief}
\end{align}
\end{lemma}
The proof is straight forward using Leibniz' formula for determinants. This lemma is frequently used in random matrix theory.

The standard Fisher-Hartwig integral can be immediately computed using this lemma. Without external matrices, i.e. when $A=D=\mathbf{1}_N$ \eqref{1FH} reads
\begin{eqnarray}
Z_{\alpha,\beta}^{\rm FH}&:=& \int  dU_N \det \Big[1+ U\Big]^\alpha\det \left[1+ U^\dagger\right]^\beta
\label{FHdef}\\
&=& c_N^{''}\prod_{l=1}^N\int_0^{2\pi}d\theta_l (1+t_l)^\alpha(1+t_l^{-1})^\beta 
\detjk\left[t_j^{k-1}\right]\detjk\left[t_j^{-k+1}\right] \nonumber\\
&=& c_N^{''}N!\detjk\left[ I_{\alpha,\beta}^{\rm FH}(j-k)\right],
\nonumber
\end{eqnarray}
where for the existence of the integral we have to require  $\Re(\alpha+\beta)>-1$, and we defined 
\begin{align}
I_{\alpha,\beta}^{\rm FH}(m):=\int_0^{2\pi}d\theta \left(1+e^{i\theta}\right)^\alpha\left(1+e^{-i\theta}\right)^\beta\exp[i(j-k)\theta].
\label{IFHdef}
\end{align}
We have used that the Jacobian of the transformation to the eigenvalues $t_j=e^{i\theta_j}$, $j=1,\ldots,N$, of the unitary matrix $U\in U(N)$  is known to be given by the modulus square of the Vandermonde determinant \eqref{Vdet}, $|\Delta_N(t)|^2$, times the constant $c_N^{''}=1/(N!(2\pi)^N)$ (recall that our convention is $\int dU_N=1$).

For complex $\alpha$ and $\beta$ the factors $(1+t_l)^\alpha(1+t_l^{-1})^\beta$ touch the branch cuts beginning at $t=-1$. However, this is not a problem in the 
expansion of the binomial series \eqref{biS}, which is absolutely convergent on the closed unit disk in the case that $\Re (\alpha ), \Re (\beta) >0$ ,
\begin{align}
  \left(1+ t \right)^\alpha \left(1+ t^{-1} \right)^\beta t^{m-1} = \sum_{n, \ell=0}^\infty \binom{\alpha}{n} \binom{\beta}{\ell} t^{n-\ell+m-1 } .
\end{align}
The integral \eqref{IFHdef} can be integrated term wise, mapping it to an integral over the unit circle $S_1= \{ z \in \Cset : |z|=1 \} $
\begin{align}
  I_{\alpha , \beta}^{\rm FH}(m) = \oint_{S_1} \frac{d t}{2 \iunit \pi} \left(1+ t \right)^\alpha \left(1+ t^{-1} \right)^\beta t^{m-1} \nonumber
  = \sum_{\ell= \max (0, m)}^\infty \binom{\alpha}{\ell -m} \binom{\beta}{\ell},
\end{align}
and applying Cauchy's theorem. Here, we could deform the circle $S_1$ such that it does not contain the endpoints of the cuts, without changing the result. 
The condition on the maximum can be removed to start the sum at $\ell=0$, since for $\Re (\alpha ) > 0$ and $\Re (\beta ) > 0$ the terms with $\ell-m<0$ vanish. Together with \eqref{ApendixA} from Lemma~\ref{ApendixATheorem} above, the remaining sum simplifies, and we obtain the following lemma, after continuing in $\alpha$ and $\beta$ to $\Re(\alpha+\beta)>-1$.
\begin{lemma}[Fisher-Hartwig Integral]\label{LaFH} For $\Re(\alpha+\beta)>-1$,  \eqref{FHdef} reads
\begin{align}
Z_{\alpha,\beta}^{\rm FH}&=\frac{1}{(2\pi)^N}\detjk\left[\binom{\alpha + \beta}{\beta -j +k} \right] 
\label{sFHresult}\\
  &=\frac{G\left(\alpha+\beta+N+1 \right)G(\alpha+1)G\left(N+1 \right) G(\beta+1)}{(2\pi)^NG(\alpha+\beta+1)G\left(\alpha+N+1 \right)G\left(\beta+N+1 \right)}.
  \label{ZFHSolution2}
\end{align}
\end{lemma}
This result is well known, see e.g. \cite{BS}, including its expression in terms of the Barnes $G$-function defined in  \cite[\S 5.17]{NIST}. 
\end{appendix}

%%%%%%%%%%%%%%%%%%%%%%%%%%%%%%%%%%%%%%%%%%%%%%%%%%%%%%

\end{document}